\documentclass[10pt,journal,compsoc]{IEEEtran}

\ifCLASSOPTIONcompsoc

\usepackage[nocompress]{cite}
\else
\usepackage{cite}
\fi
\usepackage{longtable}
\usepackage{booktabs}
\usepackage{amsmath}
\usepackage{graphicx}
\usepackage{mwe}
\usepackage{graphbox}
\usepackage{amsfonts} 
\usepackage{enumitem}
\usepackage{amssymb}
\usepackage{ragged2e}
\usepackage[colorlinks, linkcolor=blue, anchorcolor=blue, citecolor=blue]{hyperref}%为参考文献设置超链接
\usepackage[numbers,sort&compress]{natbib}%针对参考文献合并显示
\usepackage{comment}
\usepackage{color}
\usepackage{bbding} %首先在导言区调用bbding包
\usepackage{diagbox} % 加载宏包
\usepackage{stfloats}
\usepackage{float}
\usepackage{multirow}
\usepackage{subfigure}
\usepackage[justification=centering]{caption}
\usepackage{amsthm}

\usepackage{subcaption}

\usepackage[export]{adjustbox}% for the second solution
\theoremstyle{definition}

\newtheorem{theorem}{Theorem}[section]

\newtheorem{definition}[theorem]{Definition} 
\newtheorem{corollary}{Corollary}[section]
\theoremstyle{remark}

\ifCLASSINFOpdf

\ifCLASSINFOpdf

\else

\fi

\begin{document}

\title{GQHAN: A Grover-inspired Quantum Hard Attention Network}

%% 作者信息是为了躲避查重隐去的
\author{Ren-Xin Zhao,~\IEEEmembership{Member,~IEEE,} Jinjing Shi\textsuperscript{*},~\IEEEmembership{ Member,~IEEE,} and~Xuelong~Li,~\IEEEmembership{Fellow,~IEEE}
\IEEEcompsocitemizethanks{\IEEEcompsocthanksitem Ren-Xin Zhao is with the School of Computer Science and Engineering, Central South Univerisity, China, Changsha, 410083. Jinjing Shi is with the School of Electronic Information, Central South Univerisity, China, Changsha, 410083. 
\IEEEcompsocthanksitem Xuelong Li is with the Institute of Artificial Intelligence (Tele AI), China Telecom Corp Ltd, 31 Jinrong Street, Beijing 100033, P. R. China.
\IEEEcompsocthanksitem{Jinjing Shi is the corresponding author.} 
\IEEEcompsocthanksitem{E-mails: renxin\_zhao@alu.hdu.edu.cn, shijinjing@csu.edu.cn, li@nwpu.e
	du.cn.} 
}% <-this % stops an unwanted space
\thanks{Manuscript received XXXX; revised XXXX.}}
%%到这里位置

%% 以下两句话是为了躲避查重隐去的
%\markboth{IEEE TRANSACTIONS ON PATTERN ANALYSIS AND MACHINE INTELLIGENCE, XXX, XXX}{Shell \MakeLowercase{\textit{et al.}}: Bare Demo of IEEEtran.cls for Computer Society Journals}

\IEEEtitleabstractindextext{%
\begin{abstract}\justifying
Numerous current Quantum Machine Learning (QML) models exhibit an inadequacy in discerning the significance of quantum data, resulting in diminished efficacy when handling extensive quantum datasets.
Hard Attention Mechanism (HAM), anticipated to efficiently tackle the above QML bottlenecks, encounters the substantial challenge of non-differentiability, consequently constraining its extensive applicability.
In response to the dilemma of HAM and QML, a Grover-inspired Quantum Hard Attention Mechanism (GQHAM) consisting of a Flexible Oracle (FO) and an Adaptive Diffusion Operator (ADO) is proposed.
Notably, the FO is designed to surmount the non-differentiable issue by executing the activation or masking of Discrete Primitives (DPs) with Flexible Control (FC) to weave various discrete destinies.
Based on this, such discrete choice can be visualized with a specially defined Quantum Hard Attention Score (QHAS).
Furthermore, a trainable ADO is devised to boost the generality and flexibility of GQHAM.
At last, a Grover-inspired Quantum Hard Attention Network (GQHAN) based on QGHAM is constructed on PennyLane platform for Fashion MNIST binary classification.
Experimental findings demonstrate that GQHAN adeptly surmounts the non-differentiability hurdle, surpassing the efficacy of extant quantum soft self-attention mechanisms in accuracies and learning ability. In noise experiments, GQHAN is robuster to bit-flip noise in accuracy and amplitude damping noise in learning performance.
Predictably, the proposal of GQHAN enriches the Quantum Attention Mechanism (QAM), lays the foundation for future quantum computers to process large-scale data, and promotes the development of quantum computer vision.
\end{abstract}
\begin{IEEEkeywords}\justifying
Machine learning, Quantum machine learning, Grover's algorithm, Hard attention mechanism, Grover-inspired quantum hard attention Network, Quantum neural network, Quantum circuit.
\end{IEEEkeywords}}

\maketitle

\IEEEdisplaynontitleabstractindextext

\IEEEpeerreviewmaketitle

%\IEEEraisesectionheading{}
\section{Introduction}\label{introduction}

\IEEEPARstart{I}{n} recent years, QML has developed tremendously \cite{0.0,0.1,0.2}. However, many current QML models treat each quantum data equally and neglect the value of the intrinsic connections between data. This not only mandates substantial quantum storage resources for comprehensive information retention, but also constitutes a threat to the prospective large-scale quantum data processing on forthcoming quantum computers. The above urgent issues can be effectively addressed by HAM. 

HAM was first proposed in 2015 \cite{0.3}, as a discrete competitive model that adheres to the winner-take-all law which only concentrates on the most vital parts and overlooks the rest. This unique computing mechanism decreases the cost of data acquisition \cite{0.4} for designing interpretable \cite{0.5}, computationally effective \cite{0.6} and scalable \cite{0.7} models, mitigating catastrophic forgetting to a certain extent \cite{0.8}, consequently resulting in widespread applications in computer vision \cite{0.9}, natural language processing \cite{0.91} and video processing \cite{0.10,0.11}, etc. One of the striking cases is that  HAM achieves two impressive accuracies of 95.83\% and 99.07\% for safe driving recognition and driver distraction detection, respectively, in addition to a 38.71\% reduction in runtime \cite{0.12}. Although very effective, HAM is stuck in a non-differentiable dilemma that hinders its optimization due to the discrete nature of information selection processing \cite{0.13}. To conquer this dilemma, various strategies such as reinforcement learning \cite{0.91,0.10} and Gumbel-softmax based straight-through estimator \cite{0.101,0.102} have been proposed. Nevertheless, formulating a proficient reward function aligned with the task in reinforcement learning proves challenging, and misguided reward specifications may engender suboptimal attention patterns. Straight-through estimator schemes may also require a trade-off between computational efficiency and accuracy. Hence, it is inherently more practical and convenient to make HAM compatible with gradient systems. To achieve this goal, a quantum scheme inspired by Grover's algorithm is attempted to compensate for the above shortcoming.

Grover's algorithm is a quantum algorithm composed of an oracle and a diffusion operator for unstructured search \cite{0.14, 0.141}, where the oracle selects specific target items through phase flipping. The diffusion operator amplifies the amplitudes of the chosen ones and suppresses the amplitudes of the other items. Obviously, the working mechanism of Grover's algorithm is somehow similar to HAM, as it also involves discretely selecting specific targets.  However, the application of Grover's algorithm alone is insufficient to surmount the non-differentiable nature. Simultaneously, when the number of target terms is unknown or dynamically changing, the performance and adaptivity of Grover's algorithm are constrained by the inability to ascertain the appropriate number of iterations \cite{0.15,0.16,0.17,0.18}.  Thus a dramatic modification of this algorithm is necessary to create a differentiable quantum hard attention mechanism, which induces the three main propositions of this paper: 

1. Can current QML models be equipped with a quantum HAM to distinguish the intrinsic importance of quantum data?

2. How can a differentiable GQHAN be constructed by combining the Gover's algorithm with HAM?

3. How does GQHAN remain efficient when the target term is unknown or dynamically altering?

To this end, the \textbf{contribution} of this paper lies in addressing these pertinent inquiries and presenting novel insights into the potential of GQHAN: 

\begin{itemize}
	\item GQHAM is posited to discern the significance of quantum data, thereby augmenting the efficacy of QML model processing.	
	\item FO and ADO are proposed as solutions for addressing the discrete non-differentiable predicament inherent in GQHAM.	
	\item Based on GQHAM, a GQHAN is constructed on PennyLane for Fashion MNIST binary classification experiments, achieving an impressive accuracy of no less than 98\% and lower convergence values, which indicates that it has surpassed the two quantum soft attention mechanisms in terms of classification accuracy and learning ability.
\end{itemize}

\begin{table*}[ht]
	%
	%\vspace{0em}
	\def\tablename{Tab.}
	%\vspace{-6em}
	%\renewcommand{\arraystretch}{1.5}
	\centering
	\caption{Quantum Gates}
	\label{Notations}
	\small
	\begin{tabular}{@{}lccc@{}}
		\toprule
		Name of Quantum Gate & Mathematical Notation & Matrix Representation & Symbol of Quantum Gate \\ \midrule
		
			Pauli X gate& $X$& $\left[ \begin{matrix}
			0   &1 \\ 
			1   &0  \\
		\end{matrix} \right]$
		
		&         
		\includegraphics[align=c,scale=0.126]{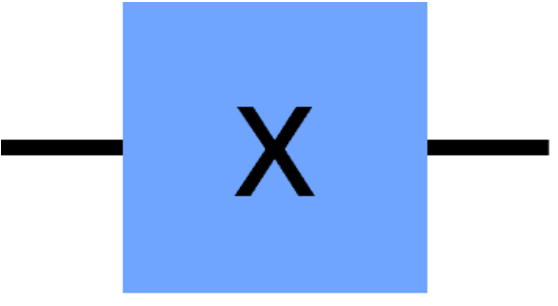} \\ \rule{-2pt}{22pt}
		
		Pauli rotating X gate& $R_{X}(\theta)$& $\left[ \begin{matrix}
 \cos(\theta/2)   & -i\sin(\theta/2) \\ 
-i\sin(\theta/2)   & \cos(\theta/2)  \\
		\end{matrix} \right]$
		
		&         
		\includegraphics[align=c,scale=0.126]{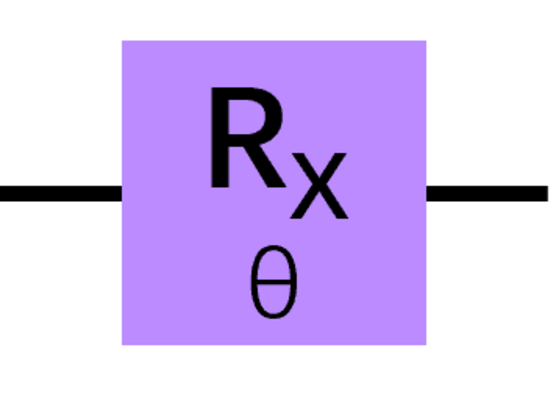} \\ \rule{-2pt}{22pt}
	
		Hadamard gate &    $H$ &$\frac{1}{\sqrt{2}}\left[ \begin{matrix}
		1 & 1  \\
		1 & -1  \\
	\end{matrix} \right]$
	&         \includegraphics[align=c,scale=0.1]{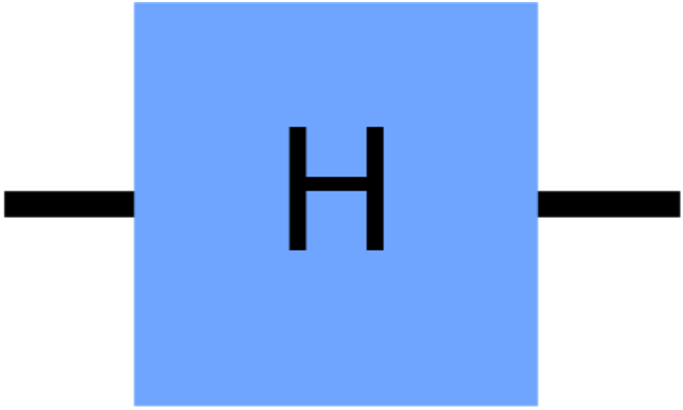}                                             	 \\  \rule{-2pt}{36pt}
	
	Controlled Y gate&    $CR_Y(\theta)$ &$\left[ \begin{matrix}
		1 & 0 & 0 & 0  \\ 
		0 & 1 & 0 & 0  \\ 
		0 & 0 & \cos(\theta/2)   & -\sin(\theta/2)    \\ 
		0 & 0 & \sin(\theta/2)  & \cos(\theta/2)    \\
	\end{matrix} \right]$
	&         \includegraphics[align=c,scale=0.13]{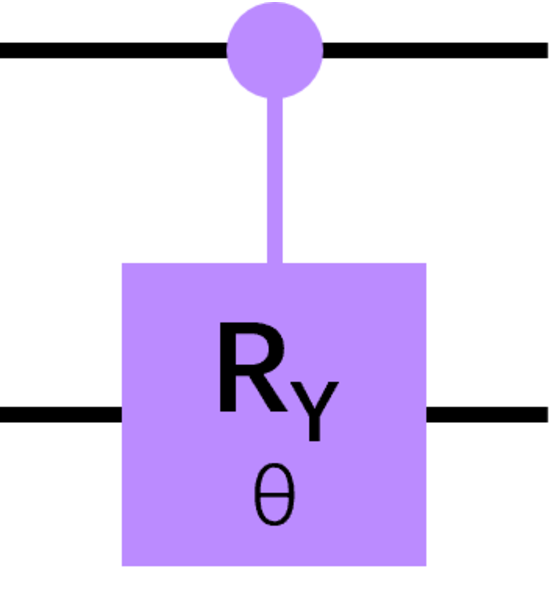}                                             	 \\ \rule{-2pt}{42pt}
		
		Multi-controlled Z gate &    $MCZ$ &$\left[ \begin{matrix}
   1 & 0 & 0 & \cdots  & 0  \\
0 & 1 & 0 & \cdots  & 0  \\
\vdots  & \vdots  & \vdots  & \ddots  & \vdots   \\
0 & 0 & 0 & \cdots  & -1  \\
		\end{matrix} \right]$
		&         \includegraphics[align=c,scale=0.13]{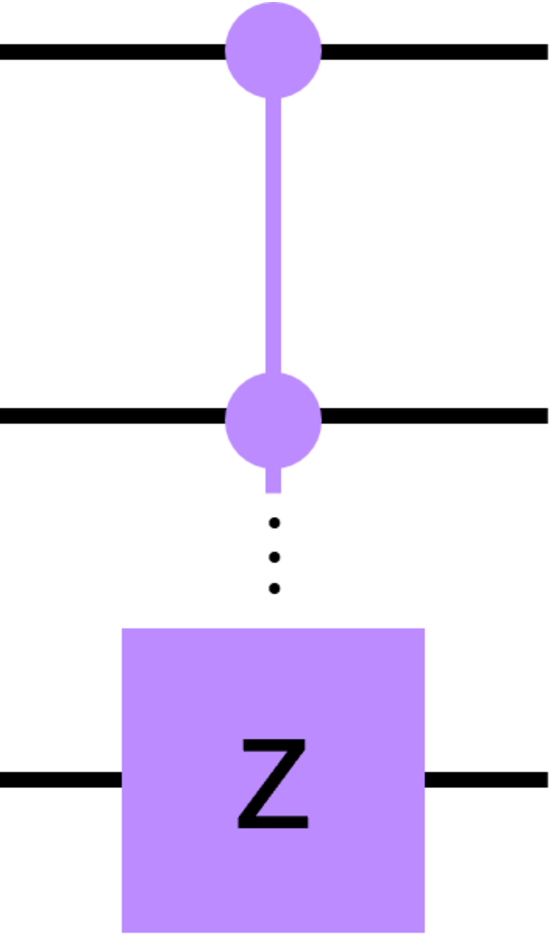}                                             	 \\ \rule{-2pt}{42pt}
		
				Discrete primitive& $C-\Lambda (b)$& Eq. (\ref{discrete primitive1})
		
		&         
		\includegraphics[align=c,scale=0.126]{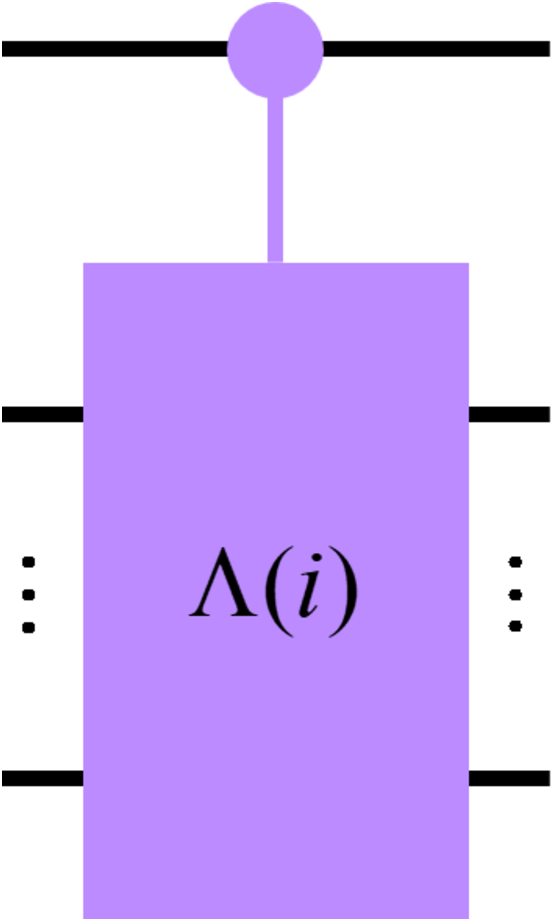} \\  
		\bottomrule
	\end{tabular}%
	%\vspace{0em}
\end{table*}
 \begin{figure*}[h]
	\centering
	\begin{minipage}[t]{1\textwidth}
		\centering\includegraphics[scale=0.15]{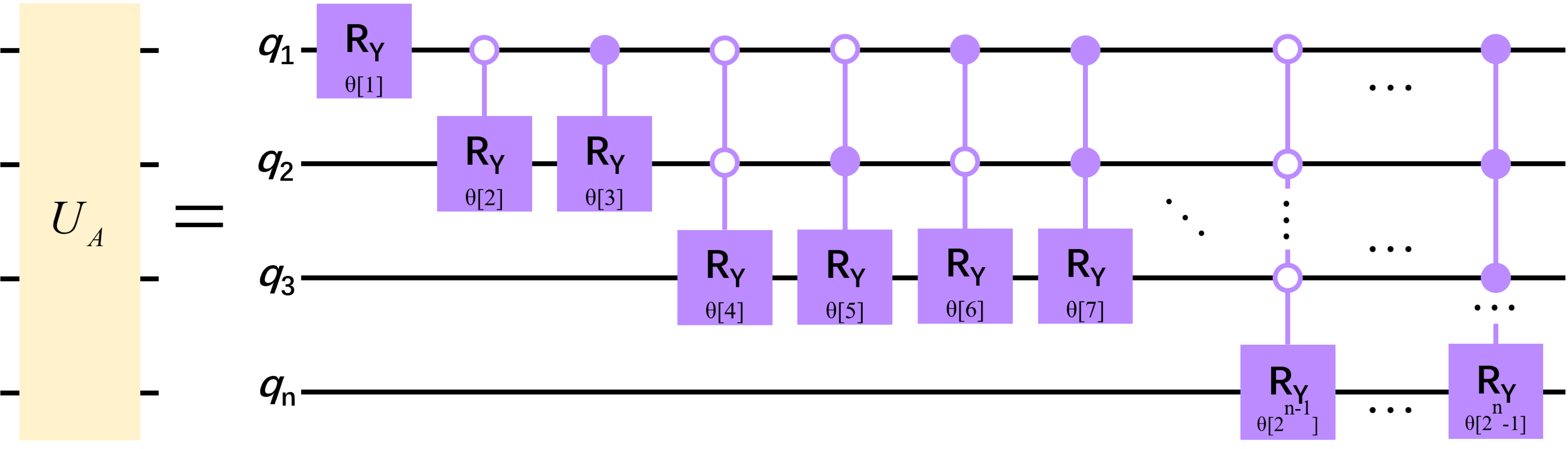}\caption{Quantum Amplitude Encoding \cite{2.0}}\label{AMPLITUDU}
	\end{minipage}
	\quad
	\begin{minipage}[t]{1\textwidth}
		\centering\includegraphics[scale=0.15]{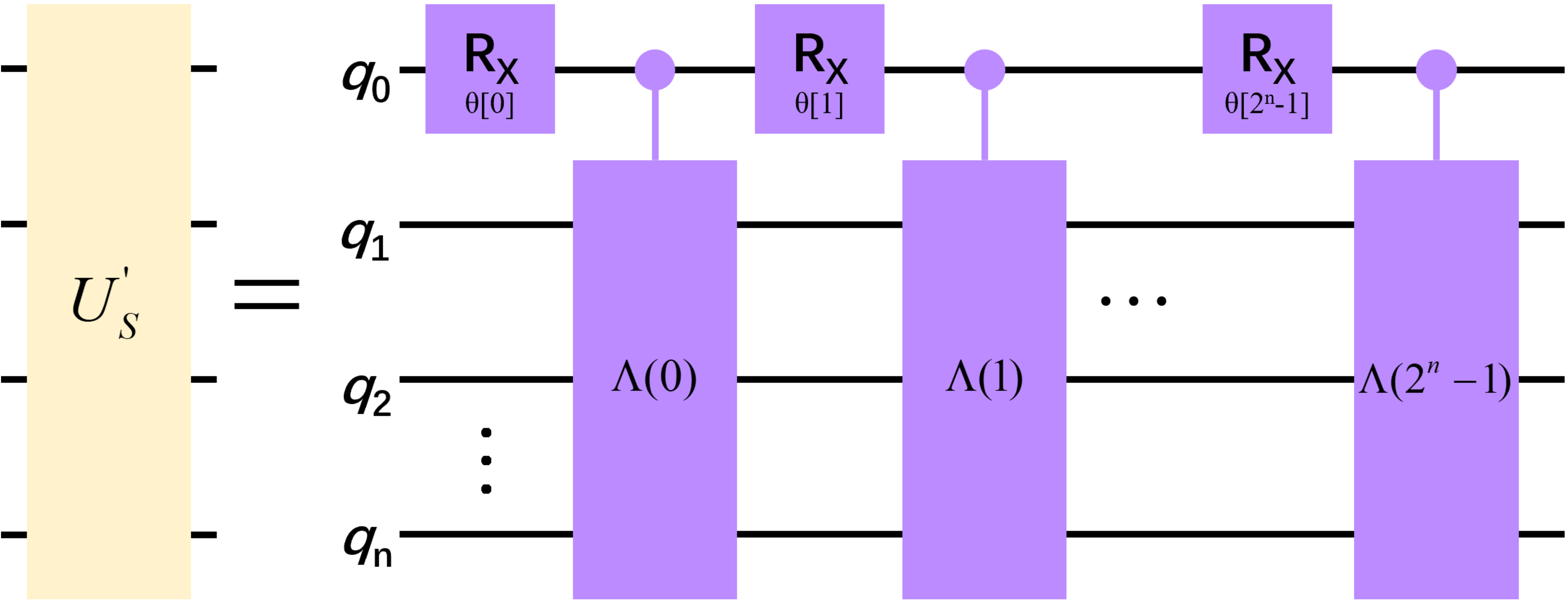}\caption{Flexible Oracle }\label{US_1}
	\end{minipage}
	\quad
	\begin{minipage}[t]{1\textwidth}
		\centering\includegraphics[scale=0.15]{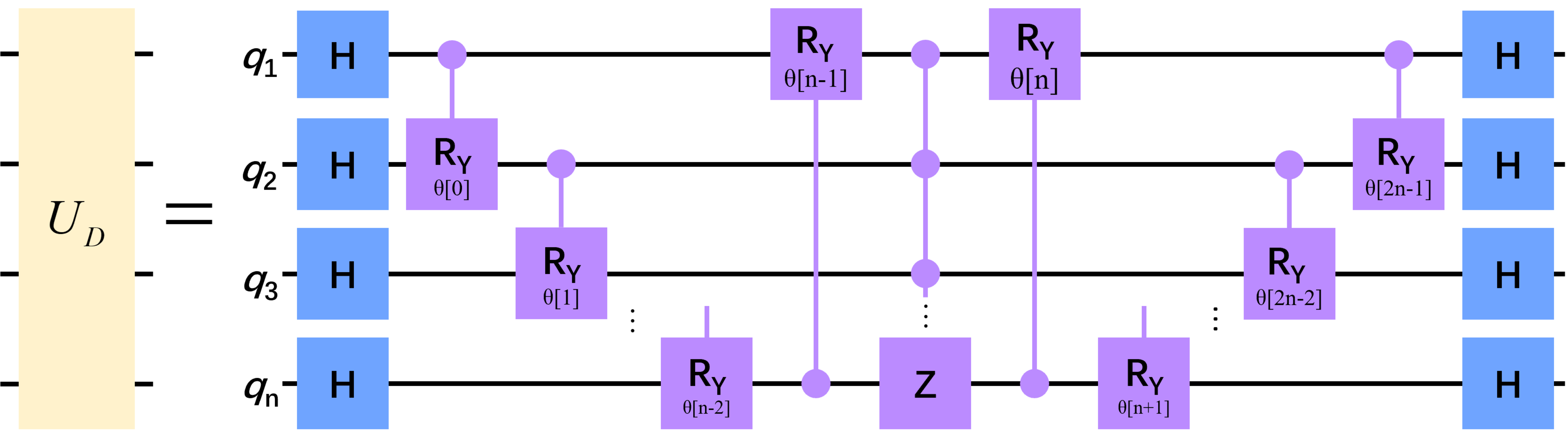}\caption{ Adaptive Diffusion Operator}\label{UD1}
	\end{minipage}
\end{figure*}

The subsequent sections of this paper are organized as follows: Section 2 contains a brief overview of QAM, quantum computing and Grover's algorithm. In Section 3, the mathematical mechanism of GQHAM is elaborated.The GQHAN framework and its workflow are described in Section 4. Section 5 delineates the experimental details and several meaningful conclusions are drawn. Eventually, a summary is presented.

\section{Related Works}

In this section, a succinct overview of quantum computing foundation, QAM, and Grover's algorithm is presented.

\subsection{Quantum Computing Foundation}\label{Quantum gates}

Qubits and quantum gates are the figurative embodiment of quantum theory. In this context, a qubit $|\psi \rangle =\alpha |0\rangle +\beta |1\rangle $ is the smallest unit that carries information, where $|0\rangle ={{[1,0]}^{\text{T}}}$ and $|1\rangle ={{[0,1]}^{\text{T}}}$ denote the ground state and the excited state respectively. $\alpha$ and $\beta$ are amplitudes satisfying $|\alpha|^2+|\beta|^2=1$ \cite{1.4}. In stark contrast to classical computation, $|\psi \rangle$ is in a superposition of $|0\rangle$ and $|1\rangle$, thereby conferring upon it the remarkable ability for exponential data representation. Moreover, the linear evolution of qubits is contingent upon quantum gates whose mathematical essence is the unitary matrices. The quantum gates used in this paper are shown in Tab. \ref{Notations}, including Pauli X gate, Pauli rotating X gate, Hadamard gate, controlled Y gate, multi-controlled Z gate and DP, where DP is one of the innovations in this paper.
\subsection{Quantum Attention Mechanisms}\label{QAM}

In QML, QAM attempts to boost the performance of QML models by discriminating the significance of quantum data. Early researches on QAM are inspired by quantum physics principles. For example, a parameter-free QAM utilizing weak measurements was introduced in 2017 to enhance bidirectional LSTM sentence modeling, demonstrating superior performance compared to classical attention mechanism \cite{1.0}. Nonetheless, these investigations exclusively depended on classical computer simulations owing to the absence of an ansatz framework. QAM with ansatz structure has not appeared until 2022. In 2022, a probabilistic full-quantum self-attention network with an exponentially scalable self-attention representation realized completely on a quantum computer was proposed \cite{1.1}. Its learning ability outperforms hardware-efficient ansatz and QAOA ansatz in a comparable configuration. In the same year, a quantum self-attention network was designed by a Baidu team to calculate quantum self-attention scores using hardware-efficient ansatz with remarkable success in text categorization \cite{1.2}. In 2023, a quantum self-attention network based on the quantum kernel method and deferred measurement principle was introduced, which achieved impressively high accuracy with few parameters \cite{1.3}. While QAM has exhibited considerable promise in domains like quantum machine vision, its theoretical underpinnings and practical applications necessitate further augmentation owing to the constraints posed by quantum hardware.

\subsection{Grover's Algorithm}
Grover's algorithm has undergone iterative enhancements in multiple aspects such as phase and initial state optimization, since its initial proposal, where phase optimization involves strategically adjusting the phase flip in the oracle.  Illustratively, in a fixed phase rotation Grover's algorithm, the rotation phase shifts from $\pi$ to $1.91684\pi$, thereby elevating the success probability to 99.59\% \cite{1.5}.  Initial state optimization allows for arbitrary initial states, markedly amplifying the generality of Grover's algorithm. Some reasons for this are that the mean and variance of the amplitude distribution of the initial state affects the measurement of the optimal time and the maximum probability of success, but the measurement of the optimal time is independent of the initial state which is an arbitrary pure state \cite{1.6,1.7}. Although numerous efforts have been conducted to boost the efficiency of Grover's algorithm, it has not been explored to incorporate trainable parameters to further strengthen its flexibility and applicability. In other words, in this paper, an attempt is presented to convert the untrainable Grover's algorithm into a trainable QML model to distinguish it from previous improvement strategies.

\section{Grover-inspired Quantum Hard Attention Mechanism}\label{sec3}

In this section, GQHAM, which successfully mitigates the non-differentiability issue, is presented and defined as follows:
\begin{definition}[Grover-inspired Quantum	Hard Attention  Mechanism] 
	\begin{equation}\label{GQHAM}
\text{GQHAM}:={{U}_{D}}U_{S}^{'}|\mathbf{In}{{\rangle_{2} }},
	\end{equation} 
	is composed of an input quantum state $|\mathbf{In}\rangle_{2}$, an ADO ${U}_{D}$, and an FO $U_{S}^{'}$  that contains three novel concepts of DP, FC and QHAS.
\end{definition}
The tenets of the FO and the ADO are elucidated emphatically in the ensuing subsections.

\subsection{ Flexible Oracle}

The classical input is set as 
\begin{equation}\label{input0}
\mathbf{In}=[{{a}_{b}}]_{b=0}^{l-1}\in {{\mathbb{R}}^{1\times l}}\subseteq \mathcal{X},
\end{equation} 
where ${a}_{b}$ indicates an element. $l$ is the total number of elements. $\mathcal{X}$ denotes the space where the input data is located. Subsequently,  Eq. (\ref{input0}) is converted to an input quantum state
\begin{equation}\label{input}
|\mathbf{In}{{\rangle_{2} }}=\sum\limits_{b=0}^{l-1}{{{\alpha }_{b}}|b{{\rangle _{2}}}}+\sum\limits_{b=l}^{{{2}^{n}}-1}{0|b{{\rangle_{2} }}}\subseteq \mathcal{H}
\end{equation} 
by quantum amplitude encoding $U_A$ in Fig. \ref{AMPLITUDU} \cite{2.0}, where the subscript 2 indicates that Eq. (\ref{input}) is situated in the second quantum register. $|b\rangle_2$ refers to the basis vector. $\mathcal{H}$ is the Hilbert space. $n=\left\lceil {{\log }_{2}}l \right\rceil$ is the amount of qubits.  $\sum\nolimits_{b=0}^{l-1}{\alpha _{b}^{2}}=1$. ${{\alpha }_{b}}={{a}_{b}}/\sqrt{\sum\nolimits_{d=0}^{l-1}{a_{d}^{2}}}$. $a_{b}, a_{d} \in {\mathbf{In}}$. Then $m$ target basis vectors 
\begin{equation}\label{M}
\mathcal{M}={{\{|{{c}_{b}}\rangle }_{2}}\}_{b=0}^{m-1}\subseteq {{\{|b\rangle_{2} }}\}_{b=0}^{{{2}^{n}}-1}
\end{equation}  
from Eq. (\ref{input}) are randomly selected and linearly combined with the opposite values of their corresponding probability amplitudes into the target term
\begin{equation}\label{target}
|\text{focus}{{\rangle_{2} }}=\sum\limits_{b=0}^{m-1}{{{\beta }_{b}}|{{c}_{b}}{{\rangle_{2} }}}.
\end{equation} 
The formation of Eq. (\ref{target}) is carried out by an oracle 
\begin{equation}\label{USmatric}
	{{U}_{S}}=\left[ \begin{matrix}
		{{(-1)}^{f(|0\rangle )}} & {} & {} & {}  \\
		{} & {{(-1)}^{f(|1\rangle )}} & {} & {}  \\
		{} & {} & \ddots  & {}  \\
		{} & {} & {} & {{(-1)}^{f(|{{2}^{n}}-1\rangle )}} \\
	\end{matrix} \right]
\end{equation} 
which is mathematically a $2^n \times 2^n$ diagonal unitary matrix, where
\begin{equation}\label{cond}
	f(x)=\left\{ \begin{array}{*{35}{l}}
		1 & x\in \mathcal{M}  \\
		0 & else  \\
	\end{array} \right..
\end{equation}  
To be specific,
\begin{equation}\label{US}
	{{U}_{S}}|b{{\rangle_{2} }}=\left\{ \begin{array}{*{35}{l}}
		-|b{{\rangle_{2} }} & |b{{\rangle_{2} }}\in \mathcal{M}  \\
		|b{{\rangle_{2} }} & else  \\
	\end{array} \right.,
\end{equation} 
which articulates that Eq. (\ref{M}) are capable of self-labeling by phase inversion when acted on by Eq. (\ref{USmatric}). 
From a holistic point of view, if Eq. (\ref{USmatric}) is updated once, the conditional judgments in Eq. (\ref{cond}) must be performed several times. This not only falls into the trap of discrete non-differentiability but also makes the problem inaccessible because it is difficult to directly reproduce such multiple complex operations as Eq. (\ref{USmatric}) and Eq. (\ref{cond}) with quantum gates. However, from another perspective, Eq. (\ref{USmatric}) can be disassembled into 
\begin{equation}\label{USd}
{{U}_{S}}=\prod\limits_{b=0}^{{{2}^{n}}-1}{{{\lambda }_{b}}},
\end{equation} 
where any element
 \begin{equation}\label{cond3}
\!{{\lambda }_{b}}\!=\!\left[ \begin{matrix}
\!	1 & {} & {} & {} & {}  \\
\!	{} & \ddots  & {} & {} & {}  \\
\!	{} & {} & (-1)_{b}^{f(|b\rangle )} & {} & {}  \\
\!	{} & {} & {} & \ddots  & {}  \\
\!	{} & {} & {} & {} & 1  \\
\end{matrix} \right]\!=\!\left\{ \begin{array}{*{35}{l}}
\!	\Lambda (b) & \!b\in \mathcal{M}  \\
\!	I & \!else  \\
\end{array} \right.\!.
 \end{equation} 
  \begin{equation}\label{A}
\Lambda (b)=\left[ \begin{matrix}
	1 & {} & {} & {} & {}  \\
	{} & \ddots  & {} & {} & {}  \\
	{} & {} & -1_b & {} & {}  \\
	{} & {} & {} & \ddots  & {}  \\
	{} & {} & {} & {} & 1  \\
\end{matrix} \right].
 \end{equation} 
Eq. (\ref{cond3}) (or Eq. (\ref{A})) represents a diagonal unitary matrix whose $b$-th element on the diagonal is $(-1)_{b}^{f(|b\rangle )} $ (or $-1$) and the others are $1$.
All symbols $I$ collectively stand for the identity matrix throughout this article. 

Deconstructing a intricate problem through perspective shifts proves more lucid and facile than contemplating it holistically in advance, since it is now only necessary to consider how to switch the two forms of Eq. (\ref{cond3}), Eq. (\ref{A}) and $I$, with continuous parameters. Besides, it is worth noting that when ${\lambda }_{b}=I$, it is equivalent to applying no operation to the current quantum circuit. Guided by the above idea, a set of DPs
\begin{equation}\label{discrete primitives}
\left\{ C-\Lambda (b) \right\}_{b=0}^{{{2}^{n}}-1}
\end{equation} 
 is introduced.
 \begin{definition}[Discrete Primitive]\label{D1}
In Eq. (\ref{discrete primitives}), an arbitrary DP
\begin{equation}\label{discrete primitive}
C-\Lambda (b)=I\otimes |0\rangle_1 \langle 0|_1+\Lambda (b)\otimes |1\rangle_1 \langle 1|_1
\end{equation} 
as shown in Tab. \ref{Notations} is a controlled quantum gate with one control bit and $n$ target bits, which means that Eq. (\ref{A}) is applied on the first quantum register only when the controlled qubit is $|1\rangle_1$. In Eq. (\ref{discrete primitive}), $\otimes$ represents the tensor symbol. $\langle0|_1$ and $\langle1|_1$ correspond to the conjugate transpositions of $|0\rangle_1$ and $|1\rangle_1$, respectively. 
 \end{definition}
  \begin{corollary}
Mathematically, Eq. (\ref{discrete primitive}) is a unitary diagonal matrix.
 \end{corollary}
   \begin{proof}
   	\begin{equation}\label{discrete primitive1}
   		\begin{aligned}
   			C-\Lambda (b)&=I\otimes |0\rangle_1 \langle 0|_1+\Lambda (b)\otimes |1\rangle_1 \langle 1|_1 \\ 
   			& =I\otimes \left[ \begin{matrix}
   				1 & 0  \\
   				0 & 0  \\
   			\end{matrix} \right]+\Lambda (b)\otimes \left[ \begin{matrix}
   				0 & 0  \\
   				0 & 1  \\
   			\end{matrix} \right] \\ 
   			& ={{\left[ \begin{matrix}
   						I & 0  \\
   						0 & \Lambda (b)  \\
   					\end{matrix} \right]}_{{{2}^{n+1}}\times {{2}^{n+1}}}}  
   		\end{aligned}
   	\end{equation} 
satisfies $C-\Lambda (b)C-\Lambda {{(b)}^{\dagger }}=C-\Lambda {{(b)}^{\dagger }}C-\Lambda (b)=I$ and is therefore a unitary diagonal matrix, where $C-\Lambda {(b)}^{\dagger }$ is the conjugate transpose of $C-\Lambda (b)$.
 \end{proof}
In adherence to Definition \ref{D1}, all elements in Eq. (\ref{discrete primitives}) function when the controlled qubit is $|1\rangle_1$. Leveraging this uniqueness, quantum gates containing parameters are used to determine whether the controlled qubit is $|1\rangle_1$ or not, thus selectively activating and shielding certain elements in Eq. (\ref{discrete primitives}). For this purpose, the concept of FC is proposed.
\begin{definition}[Flexible Control]\label{Flexible Control}
Suppose that the quantum control bit of Eq. (\ref{discrete primitive}) is initialized to $|0\rangle_1$. Subsequently Pauli rotating X gates ${{R}_{X}(\theta)}$ are posed on $|0\rangle_1$ to form the flexible control
\begin{equation}\label{cond1}
{{R}_{X}}(\theta )|0{{\rangle_{1} }}=\left\{ \begin{array}{*{35}{l}}
	|1{{\rangle_{1} }} & \theta =(4k+1)\pi ,k\in \mathbb{Z}  \\
	{{R}_{X}}(\theta )|0{{\rangle_{1} }} & else  \\
\end{array} \right..
\end{equation}
Eq. (\ref{cond1}) shows that the control bit is $|1\rangle_1$ only when the continuous variable $\theta$ is located at $(4k+1)\pi ,k\in \mathbb{Z}$.
\end{definition}
\begin{proof}
In quantum computing, the interconversion between ground state $|0\rangle_1$ and excited state $|1\rangle_1$ can be realized by a Pauli X gate $X$.
Therefore, ${{R}_{X}(\theta)}$ is equal to the form of $X$ only when $\tfrac{\theta }{2}\text{=}\tfrac{\pi }{2}+2k\pi ,k\in \mathbb{Z}$. That is, \[\left[ \begin{matrix}
	\cos \left( \frac{\pi }{2}+2k\pi  \right) & -i\sin \left( \frac{\pi }{2}+2k\pi  \right)  \\
	-i\sin \left( \frac{\pi }{2}+2k\pi  \right) & \cos \left( \frac{\pi }{2}+2k\pi  \right)  \\
\end{matrix} \right]=-i\left[ \begin{matrix}
	0 & 1  \\
	1 & 0  \\
\end{matrix} \right],\]
where the global phase $-i$ can be neglected.
\end{proof}
Naturally, the definition of hard attention scores can be derived from Definition \ref{Flexible Control}:
\begin{definition}[Quantum Hard Attention Score]
	\begin{equation}\label{cond2}
\text{QHAS}:=\left\{ \begin{array}{*{35}{l}}
	1 & \theta =(4k+1)\pi ,k\in\mathbb{Z}   \\
	0 & else  \\
\end{array} \right.
	\end{equation}
can be used to visualize the choice of FC in training.
\end{definition}

Based on Eq. (\ref{discrete primitives}) and Eq. (\ref{cond1}), FO can be created.
\begin{definition}[Flexible Oracle]
The FO
 \begin{equation}\label{flexible oracle}
U_{S}^{'}=\prod\limits_{b=0}^{{{2}^{n}}-1}{{{R}_{X}}({{\theta }_{b}})\otimes I\times C-\Lambda (b)},
\end{equation} 
as shown in Fig. \ref{US_1}, replaces Eq. (\ref{USmatric}) by applying flexible control to each DP in Eq. (\ref{discrete primitives}), thus coming to select and combine different DPs and create new discrete operations, ultimately freeing GQHAM from the confinement of discrete non-differentiability.
\end{definition}

 \subsection{Adaptive Diffusion Operator}
 
Eq. (\ref{flexible oracle}) focuses attention on important $\mathcal{M}$, but the amplitudes of these $\mathcal{M}$ may be so small that they cannot be thoroughly sorted out, which requires amplifying their quantum amplitudes while shrinking the amplitudes of the irrelevant terms by the diffusion operator. The diffusion operator in the original Grover algorithm \cite{0.14} is defined as  
\begin{equation}\label{G}
{{U}_{D}}={{H}^{\otimes n}}{{X}^{\otimes n}}(MCZ){{X}^{\otimes n}}{{H}^{\otimes n}}.
\end{equation} 
Another geometric interpretation is that Eq. (\ref{G}) is equivalent to 
\begin{equation}\label{G1}
{{U}_{D}}=2|\mathbf{In}\rangle \langle \mathbf{In}|-I,
\end{equation} 
where $|\mathbf{In}\rangle$ is the conjugate transpose of $\langle \mathbf{In}|$. $I$ is the identity matrix. This implies that the symmetry axis of Grover's algorithm is $|\mathbf{In}\rangle$ for each iteration, which loses flexibility. Thus, Eq. (\ref{G}) is improved and the ADO is proposed.

 \begin{definition}[Adaptive Diffusion Operator] The ADO
\begin{equation}\label{G2}
{{U}_{D}}={{U}_{1}}(MCZ)U_{1}^{\dagger },
\end{equation} 
is shown in Fig. \ref{UD1}, where
\begin{equation}\label{G3}
{{U}_{1}}=\prod\limits_{b=0}^{n-1}{C{{R}_{Y}}({{\theta }_{b}})[b,f]}\underset{b=0}{\mathop{\overset{n-1}{\mathop{\otimes }}\,}}\,H[b].
\end{equation} 
\begin{equation}\label{G4}
U_{1}^{\dagger }=\underset{b=0}{\mathop{\overset{n-1}{\mathop{\otimes }}\,}}\,H[b]\prod\limits_{b=0}^{n-1}{C{{R}_{Y}}({{\theta }_{b+n}})[b,f]}.
\end{equation} 
\begin{equation}\label{G5}
f=\left\{ \begin{array}{*{35}{l}}
	c+1 & c\ne n-1  \\
	0 & c=n -1  \\
\end{array} \right..
\end{equation} 
\end{definition}
Eq. (\ref{G3}) and Eq. (\ref{G4}) use a quantum coordinate representation \cite{1.1}.
Eq. (\ref{G2}) replaces $X$ in Eq. (\ref{G}) with $CR_Y(\theta)$ to make it trainable, which indicates that this axis of symmetry is no longer fixed and has flexibility.

 At present, the GQHAM mechanism stands elucidated comprehensively. The ingenuity of GQHAM resides in the activation or masking of DPs through continuous parameters, thereby amalgamating diverse discrete choices without transgressing the tenets of gradient backpropagation. This innovation ultimately surmounts the challenge of non-differentiability. Furthermore, enhancements have been instituted in the original expansion operator to augment the flexibility of GQHAM.

\section{Grover-inspired Quantum Hard Attention Network}\label{workflow}
In this section, GQHAN in Fig. \ref{GQHAN} is constructed based on GQHAM and its workflow is elaborated.
 \begin{figure}[h]
		\centering\includegraphics[scale=0.2]{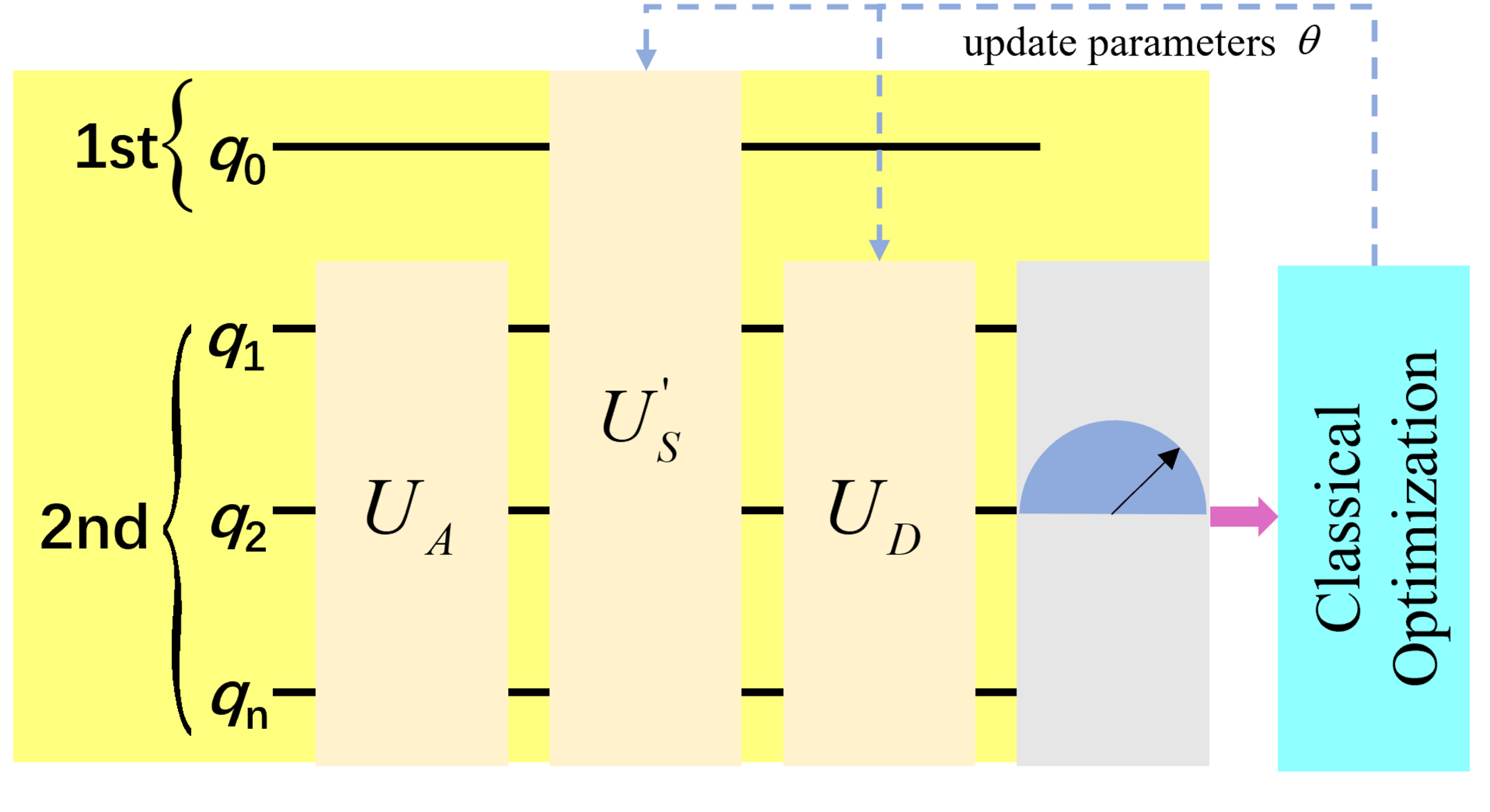}\caption{Grover-inspired Quantum Hard Attention Network}\label{GQHAN}
\end{figure}

As depicted in Fig. \ref{GQHAN}, GQHAN is segregated into dual components: the ansatz segment in the yellow box and the classical optimizer module in the blue box. Among them, the ansatz segment is constituted by $n+1$ qubits, 
encompassing an ancillary qubit $q_0$ in the first quantum register and $n$ input qubits ${{q}_{1}}\sim {{q}_{n}}$  in the second quantum register. In addition, $U_A$ in Fig. \ref{AMPLITUDU}, $U_{S}^{'}$ in Fig. \ref{US_1}  and ${U}_{D}$ in Fig. \ref{UD1} sequentially are quantum amplitude encoding, FO and ADO. Ultimately, the data is extracted via quantum measurement, depicted within the gray box in Fig. \ref{GQHAN}, and subsequently transmitted to the classical optimizer for further processing. Employing classical optimizers, such as the quantum natural simultaneous perturbation stochastic approximation optimizer \cite{3.0} or the quantum natural gradient optimizer \cite{3.1}, the parameters undergo iterative refinement until convergence of the loss function is achieved. It is noteworthy that in this paper, the Nesterov momentum optimizer \cite{3.2}, a variant amalgamating momentum components with gradient descent, is utilized to incorporate historical gradients into the optimization process.

In accordance with the framework depicted in Fig. \ref{GQHAN}, the precise workflow of GQHAN unfolds as follows.

Step 1: The initial quantum state
\begin{equation}\label{Step1}
	|\Psi {{\rangle_{1} }}\otimes |\Phi {{\rangle _{2}}}=|0\rangle_{1} \otimes |0^{\otimes n}\rangle_{2}
\end{equation}
is prepared in the first and second quantum registers.

Step 2: Eq. (\ref{input0}) are embedded into GQHAN by quantum amplitude encoding $U_A$:
\begin{equation}\label{Step2}
|\Psi {{\rangle_{1} }}\otimes |\Phi {{\rangle_{2} }}\xrightarrow{{{U}_{A}}}|0{{\rangle_{1} }}\otimes |\mathbf{In}{{\rangle_{2} }}.
\end{equation}

Step 3: Eq. (\ref{flexible oracle}) is used to adjust the continuous parameter during training to partition Eq.  (\ref{input}) into its salient and negligible components, $|\text{focus}\rangle_{2} $ and $|\overline{\text{focus}}\rangle_{2} $:
\begin{equation}\label{Step3}
|\Psi {{\rangle_{1} }}\otimes |\Phi {{\rangle _{2}}}\xrightarrow{U_{S}^{'}}|\phi{{\rangle_{1} }}\otimes {{(|\text{focus}\rangle_{2} }}+|\overline{\text{focus}}{{\rangle_{2} }}),
\end{equation}
where  $|\phi{{\rangle_{1} }}$ is determined by Eq. (\ref{cond1}).

Step 4: Apply Eq. (\ref{G2}) to dynamically amplify the amplitudes in Eq. (\ref{target}) while reducing the other amplitudes to end up as
\begin{equation}\label{Step4}
|\Psi {{\rangle_{1} }}\otimes |\Phi {{\rangle_{2} }}\xrightarrow{{{U}_{D}}}|\phi{{\rangle _{1}}}\otimes |\widetilde{\text{focus}}{{\rangle_{2} }}.
\end{equation}

Step 5: For specific Fashion MNIST binary classification, the expectation value 
\begin{equation}
	\mathbb{E}=\langle T|P|T \rangle
\end{equation} 
on the last qubit $q_n$ is measured as the predicted label, where $|T\rangle=|\Psi {{\rangle_{1} }}\otimes |\Phi {{\rangle_{2} }}$. $P$ is the projection operator.

Step 6: Cost function is
\begin{equation}
	f(\textbf{In},{{\boldsymbol{\theta }}})=\frac{1}{m}\sum\limits_{i=1}^{m}{{{[{{y}_{i}}-\text{sgn} (\mathbb{E})]}^{2}}},
\end{equation} 
where $\text{sgn}(\cdot)$ is the sign function. $m$ is the number of terms. ${y}_{i}$ stands for real label. ${{\boldsymbol{\theta }}}$ are the trainable parameters. The optimization rules for the Nesterov momentum optimizer are as follows:
\begin{equation}
	\left\{ \begin{aligned}
		& {{\boldsymbol{\theta }}^{(t+1)}}={{\boldsymbol{\theta }}^{(t)}}-{{a}^{(t+1)}} \\ 
		& {{a}^{(t+1)}}=\gamma \cdot {{a}^{(t)}}+\eta \cdot \nabla f(\mathcal{D},{{\mathcal{D}}^{\prime }},{\boldsymbol{\theta }^{(t)}}-\gamma \cdot {{a}^{(t)}}) \\ 
	\end{aligned} \right.\,
\end{equation} 
where the superscript $t$ represents the $t$-th iteration and $t+1$ stands for the next iteration. The momentum term $\gamma$ is adjustable and generally takes the value 0.9. $\eta$ is the learning rate. $a^{(t+1)}$ and $a^{(t)}$ are accumulator terms. $\nabla f(\mathcal{D},{{\mathcal{D}}^{\prime }},{\boldsymbol{\theta }^{(t)}}-\gamma \cdot {{a}^{(t)}})$ denotes the gradient.

In the end, all the above steps are repeated until the cost function converge. 
To sum up, by analyzing Step 1 to 6 above, the operation mechanism of GQHAN in the classification problem is revealed, which lays a theoretical foundation for the experiment.

\section{Experiment}\label{sec4}
In this section, GQHAN is implemented on the PennyLane platform to conduct binary classification experiments on Fashion MNIST.
Precisely, the experiments are categorized into the subsequent segments.
\begin{itemize} 
	\item The performance of GQHAN,  QSAN \cite{1.1} and QKSAN \cite{1.3} is evaluated and compared in the Fashion MNIST binary classification task under a uniform classical optimizer configuration and no noise. 
	\item A visualization of QHAS is performed to represent the results before and after training.
	\item To elucidate the performance on a real quantum computer, the impact of bit-flip error and amplitude damping error on GQHAN is scrutinized.
\end{itemize} 

%\subsection{Model}

\subsection{Dataset}

Fashion MNIST \cite{4.3}, the recognized and widely adopted benchmark datasets in machine learning, consists of 10,000 test images and 60,000 training images, respectively, each containing 28 by 28 pixel points. In this experiment, 550 images labeled 0 and another 550 tagged 1 are taken as the dataset from Fashion MNIST. 500 images from each class of tags are randomly sampled and assembled into a training set with the rest as a test set. Moreover, acknowledging the existing constraint on the number of public qubits, exemplified by the limited provision of 5 to 7 free-to-use qubits by IBM, one strategy is to compress the dimensionality of all images in the dataset to 8 using the principal component analysis algorithm.

\subsection{Experimental Setting}
 \begin{figure*}[ht]
	\centering
	\begin{minipage}[h]{1\textwidth}
		\centering\includegraphics[scale=0.24]{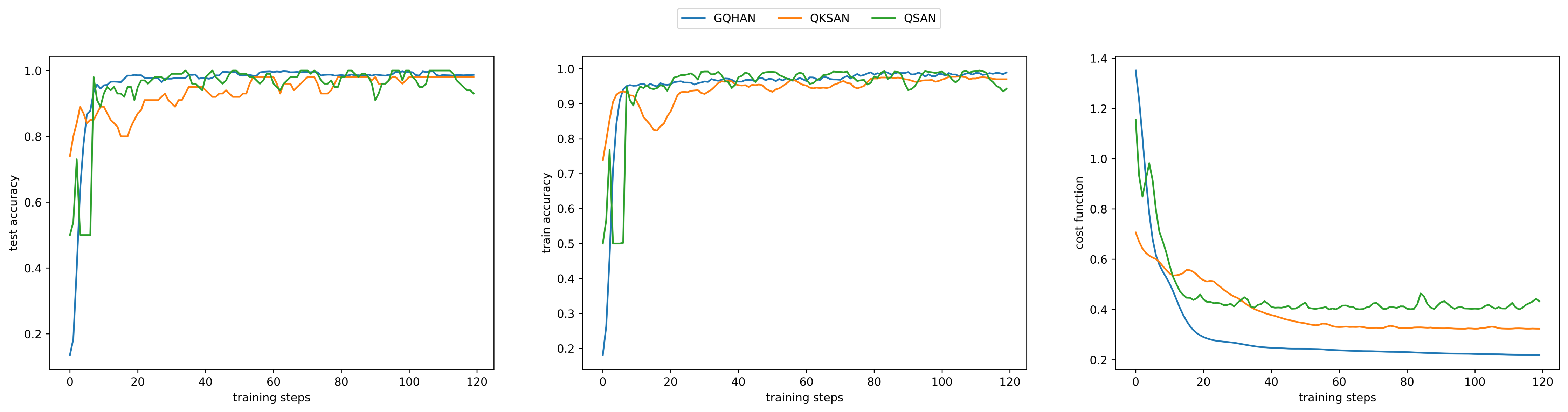}\caption{Comparison of Fashion MNIST Binary Classification for GQHAN, QKSAN and QSAN.}\label{threeModels}
	\end{minipage}
	\quad
			\begin{minipage}[t]{1\textwidth}
		\centering\includegraphics[scale=0.7]{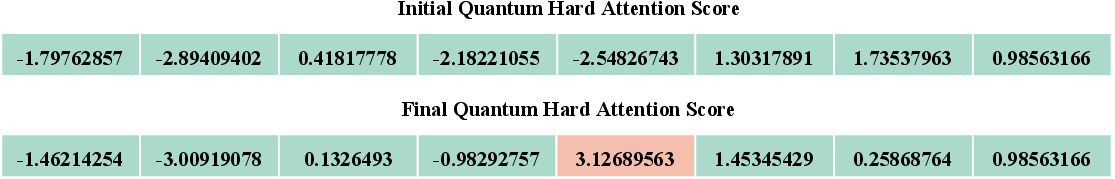}\caption{Visualization of Quantum Hard Attention Score}\label{QHAS}
	\end{minipage}
	\quad
	\begin{minipage}[t]{1\textwidth}
		\centering\includegraphics[scale=0.24]{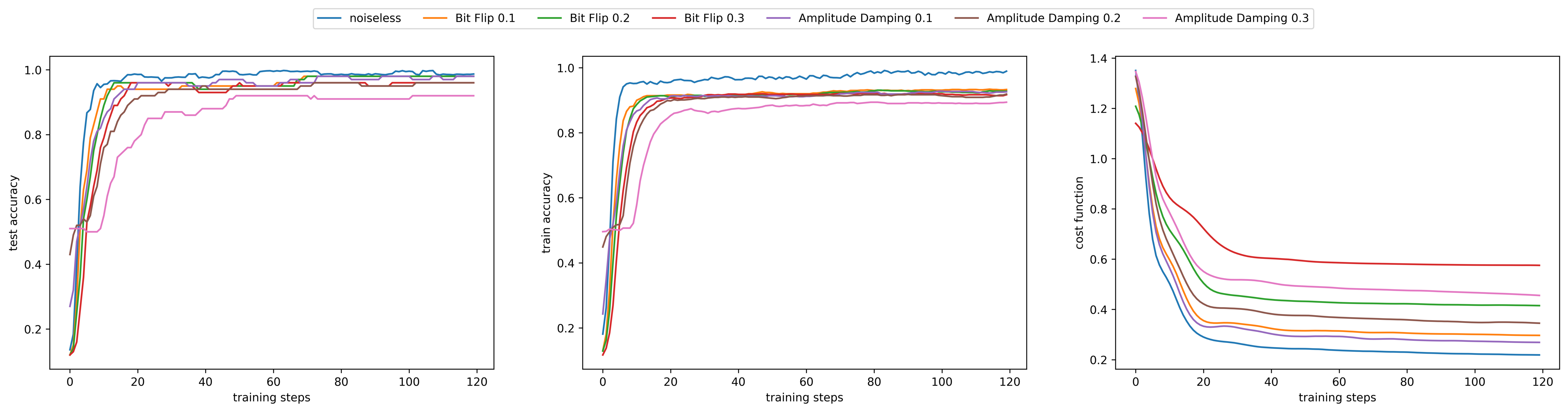}\caption{Noise Experiments of GQHAN}\label{noise}
	\end{minipage}

\end{figure*}
Tab. \ref{Experimental comparison} delineates the specific experimental configurations of GQHAN,  QSAN \cite{1.1} and QKSAN \cite{1.3}, detailing key parameters of both the ansatz and classical optimizers. 
In the ansatz, paramount metrics encompass parameter counts in trainable layers, layer quantity, and requisite qubits. Classical optimization scrutinizes factors such as learning rate, loss function type, batch size, maximum step size, optimizer type, and relevant parameters. The type and configuration of the classical optimizer, along with the number of ansatz parameters, are maintained as consistently as possible to ensure and underscore equitable comparisons between quantum models. Additionally, a substantial learning rate is deliberately chosen to expedite convergence.
\begin{table}[h]
	\small 
	\centering
	\def\tablename{Tab.}
	\caption{Experimental Configuration}
	\label{Experimental comparison}
	\begin{tabular}{lccc}
		
		\toprule
		%\centering
		%\caption{Experimental configuration}
		\multirow{2}{*}{Indicators} & \multicolumn{3}{c}{Models} \\
		\cmidrule(lr){2-4}
		&  GQHAN & QKSAN \cite{1.3}  & QSAN \cite{1.1}\\
		\midrule
		parameters                      & 14               & 14   &9      \\ 
		
		layers                     & 7               & 6      & 26   \rule{0pt}{12pt} \\  
		
		qubits & \multicolumn{2}{c}{4}                    & 8                         \rule{0pt}{12pt}    \\
		
		learning rate & \multicolumn{3}{c}{0.09}                                               \rule{0pt}{12pt}         \\
		
		loss function & \multicolumn{3}{c}{square loss}                                            \rule{0pt}{12pt}            \\
		
		batch\_size   & \multicolumn{3}{c}{30}                                                 \rule{0pt}{12pt}          \\
		
		step          & \multicolumn{3}{c}{120 steps}                                               \rule{0pt}{12pt}          \\
		
		optimizer     & \multicolumn{3}{c}{Nesterov Momentum Optimizer: $\gamma = 0.9$}              \rule{0pt}{12pt}                       \\
		\bottomrule
	\end{tabular}
	%\end{wraptable}
\end{table}  
\begin{table}[h]
	%\normalfont
	\centering
	\def\tablename{Tab.}
	\caption{Critical Data}
	\label{data}
	\resizebox{\columnwidth}{!}{%
		\begin{tabular}{lccc@{}}
			\toprule
			\multicolumn{1}{l}{\multirow{2}{*}{Indicators}} & \multicolumn{3}{c}{Models} \\ \cmidrule(l){2-4} 
			&  GQHAN & QKSAN \cite{1.3}  & QSAN \cite{1.1}  \\ \midrule
			test accuracy of the last 10 steps                      & 98.59\%              & 98\%   &96.8\%       \\ 
			train accuracy of the last 10 steps                      & 98.65\%              & 97.22\%   &96.77\%     \rule{0pt}{12pt}   \\ 
			number of step to start convergence                     & 19              & 25      & 55   \rule{0pt}{12pt}  \\  
			convergence value of the loss function                     & 0.219             & 0.323      & 0.42   \rule{0pt}{12pt}   \\  
			\bottomrule
		\end{tabular}%
	}
\end{table}

\subsection{Experimental Analysis}

\subsubsection{Classification Experiment}

The outcomes of the comparative analysis involving GQHAN and two quantum soft self-attention mechanisms, namely QKSAN and QSAN, are depicted in Fig. \ref{threeModels}.
In Fig. \ref{threeModels}, the abscissae of all three subplots represent the training steps, while the ordinates denote the test accuracy, train accuracy, and loss function sequentially. 
The blue, orange and green curves delineate GQHAN, QKSAN, and QSAN correspondingly.
According to Fig. \ref{threeModels}, the critical data are enumerated in Tab. \ref{data}.
Finally, the following conclusions can be drawn.

\begin{itemize}
	\item By averaging the accuracy over the last 10 steps, GQHAN attains a test accuracy of 98.59\%, surpassing QKSAN at 98\% and QSAN at 96.8\%. Regarding training accuracy, it outperforms QKSAN by 1.42\% and QSAN by 1.87\%.
	\item In terms of convergence speed, GQHAN initiates convergence around step 20, slightly lagging behind QSAN by approximately 5 steps but significantly outpacing QKSAN.
	\item Concerning convergence values, GQHAN converges to 0.219, compared to 0.323 for QSAN and 0.42 for QKSAN, indicating a superior learning capability.
\end{itemize}

In summary, the analysis and comparison from three perspectives demonstrates that GQHAN is able to slightly surpass the two quantum soft self-attention mechanisms in terms of performance in the Fashion MNIST classification.
\subsubsection{Visualization of Quantum Hard Attention Score}
The visualization results of QHAS under the experimental conditions of the previous subsection are shown in Fig. \ref{QHAS}. All values in Fig. \ref{QHAS} represent the parameters $\theta$ in $R_X(\theta)$. Here, green color is used to indicate that they are not selected, while red color has the opposite meaning of green color. Since the parameter initialization are random, this batch of parameters does not meet the requirement of Eq. (\ref{cond2}). Instead, in the last round of training, the parameter 3.12689563 is marked red and is very close to $\pi$, at which point it can be considered as being discretely selected in engineering.

\subsubsection{Noise Experiments}

The amplitude damping and bit-flip noise experiments for GQHAN are illustrated in Fig. \ref{noise}. \begin{table}[h]
	%\normalfont
	\centering
	\def\tablename{Tab.}
	\caption{Critical Data for Noise Experiments}
	\label{data1}
	\resizebox{\columnwidth}{!}{%
		\begin{tabular}{lccc@{}}
			\toprule
			\multicolumn{1}{l}{\multirow{2}{*}{Models}} & \multicolumn{3}{c}{Indicators} \\ \cmidrule(l){2-4} 
			& test accuracy    & train accuracy  & convergence value  \\ \midrule
			noiseless                     & 98.59\%              & 98.65\%   & 0.219       \\ 
			amplitude damping 0.1                      & 97.6\%              & 92.42\%   &0.269    \rule{0pt}{12pt}   \\ 
			amplitude damping 0.2                     & 96\%              & 91.24\%      & 0.348   \rule{0pt}{12pt}  \\  
			amplitude damping 0.3                    & 92\%            & 89.14\%    & 0.459   \rule{0pt}{12pt}   \\  
			bit flip 0.1                      & 98\%              & 93.27\%   &0.297    \rule{0pt}{12pt}   \\   
			bit flip 0.2                   & 98\%              & 92.68\%      & 0.416   \rule{0pt}{12pt}  \\  
			bit flip 0.3                     & 96\%            & 91.58\%    & 0.576   \rule{0pt}{12pt}   \\  
			\bottomrule
		\end{tabular}%
	}
\end{table}Except for the blue line which indicates the experimental results in the absence of noise, Fig. \ref{noise} depicts three sets of experiments conducted for both amplitude damping and bit-flip noise, each occurring with probabilities of 0.1, 0.2, and 0.3. Likewise, the mean metrics derived from the final 10 steps upon completion of the training regimen were employed as pivotal assessment data, as delineated in Tab. \ref{data1}. 
According to the data presented in Tab. \ref{data1}, discernible conclusions can be derived.
\begin{itemize}
	\item In assessing the diminution of test and training accuracy, the impact of amplitude damping noise exhibits a more conspicuous manifestation.
	\item With equiprobability of noise emergence, the bit-flip perturbation induces a more pronounced escalation in convergence values, thereby exacerbating the learning efficacy degradation of GQHAN.
\end{itemize}

In total, GQHAN exhibits a degree of resilience against both types of noise. Concerning accuracy, the impact of bit-flip noise manifests as a comparatively minor decline in intuitive values. However, it is noteworthy that the degradation in its effect on learning performance is more conspicuous.

\section{Conclusion}\label{sec5}
A quantum HAM, GQHAM, is proposed to compensate for the inability of many current QMLs to recognize the significance of quantum data. In addition, FO and ADO are designed to overcome the nativity of non-differentiability due to discrete selection. A QHAS visualization definition is also providedon this premise. Ultimately, GQHAN, which can be deployed on a quantum computer, is constructed based on GQHAM. In binary classification experiments on Fashion MNIST on the PennyLane platform, GQHAN achieves 98.59\% and 98.65\% test and train accuracies, respectively, and outperforms the two quantum soft self-attention in terms of learning ability. In the noise experiments, contrasted with amplitude damping noise, the impact of bit flip noise on GQHAN is relatively diminished in precision but looms larger in learning capabilities. Our approach contributes to the QML model to pay more attention to the important parts of quantum data, laying the foundation for future quantum computers to process massive amounts of high-dimensional data.
\ifCLASSOPTIONcompsoc

%\section*{Acknowledgment}

\begin{IEEEbiographynophoto}{Ren-Xin Zhao}(Member, IEEE) received his B.S. degree from the College of Automation, Hangzhou Dianzi University, Hangzhou, China in 2017 and his M.S. degree from the College of Electrical and Information Engineering, Hunan University, Changsha, China in 2020. He is now a PhD student in the  School of Computer Science and Engineering at Central South University, Changsha, China. His interests include quantum machine learning, quantum neural networks, and design and optimization of quantum circuits.\end{IEEEbiographynophoto}

\begin{IEEEbiography} [{\includegraphics[width=1in,height=1.25in,clip,keepaspectratio]{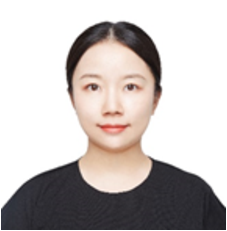}}]{Jinjing Shi}(Member, IEEE) is now a professor in the School of Electronic Information of Central South University. She received her B.S. and Ph.D. degrees in the School of Information Science and Engineering, Central South University, Changsha, China, in 2008 and 2013, respectively. She was selected in the ”Shenghua lieying” talent program of Central South University and Special Foundation for Distinguished Young Scientists of Changsha in 2013 and 2019, respectively. Her research interests include quantum computation and quantum cryptography. She has presided over the National Natural Science Foundation Project of China and that of Hunan Province. There are 50 academic papers published in important international academic journals and conferences. She has received the second prize of natural science and the outstanding doctoral dissertation of Hunan Province in 2015, and she has received the Best Paper Award in the international academic conference MSPT2011 and Outstanding Paper Award in IEEE ICACT2012.\end{IEEEbiography}

\begin{IEEEbiographynophoto}{Xuelong Li} (M'02-SM'07-F'12) is with the Institute of Artificial Intelligence (Tele AI), China Telecom Corp Ltd, 31 Jinrong Street, Beijing 100033, P. R. China
\end{IEEEbiographynophoto}

%We would like to thank all the reviewers who provided valuable suggestions.

\ifCLASSOPTIONcaptionsoff
  \newpage
\fi

\end{document}